%% file: Template.tex
\documentclass{article}

\usepackage{spconf,amsmath,graphicx}
\usepackage{amssymb}
\usepackage{amsthm}
\usepackage{amsmath}
\usepackage{algorithm}
\usepackage{algorithmic}
\usepackage{booktabs}
\usepackage{wrapfig}
\usepackage[absolute]{textpos}

\newtheorem{theorem}{Theorem}

\newtheorem{remark}{Remark}

\title{
LATENT-SHIFT: GRADIENT OF ENTROPY HELPS NEURAL CODECS
}
\name{Muhammet Balcilar, Bharath Bhushan Damodaran, Karam Naser, Franck Galpin \& Pierre Hellier \thanks{© 2023 IEEE. Personal use of this material is permitted. Permission from IEEE must be obtained for all
other uses, in any current or future media, including reprinting/republishing this material for advertising or
promotional purposes, creating new collective works, for resale or redistribution to servers or lists, or reuse
of any copyrighted component of this work in other works.}}
\address{InterDigital, Inc., France}
\begin{document}

\maketitle

%
\begin{abstract}
End-to-end image/video codecs are getting competitive compared to traditional compression techniques that have been developed through decades of manual engineering efforts.
These trainable codecs have many advantages over traditional techniques such as easy adaptation on perceptual distortion metrics and high performance on specific domains thanks to
their learning ability. However, state of the art neural codecs does not take advantage of the existence of gradient of entropy in decoding device. In this paper, we theoretically show that gradient of entropy (available at decoder side) is correlated with the gradient of the reconstruction error (which is not available at decoder side). We then demonstrate experimentally that this gradient can be used on various compression methods, leading to a  $1-2\%$ rate savings for the same quality. Our method is orthogonal to other improvements and brings independent rate savings. 


\end{abstract}

\input{introduction}

\input{problem_statement}

\input{proposal}

\input{results}
\input{conclusion}
\clearpage
\bibliographystyle{IEEEbib}
\bibliography{strings,refs}

\end{document}

%% file: introduction.tex
\section{Introduction}
\label{sec:intro}

Lossy Image and video compression is a fundamental task in image processing. 
Thanks to the community's decades long efforts, traditional methods (e.g. VVC) have reached current state of the art rate-distortion (RD) performance and dominate current codecs market. Recently, end-to-end trainable deep models have emerged with promising RD performances by learning the informative latents and modeling the latent distribution. 

End-to-end deep compression methods are generally rate-distortion autoencoders \cite{habibian2019video}, where the latents are optimized using a rate-distortion loss function. For perceptual friendly compression, distortion based on a perceptual metric can also be used in the loss function \cite{blau19a}. These methods can be seen as a special case of Variational Autoencoder (VAE) models as described in \cite{kingma2013auto}, where the approximate posterior distribution is a uniform distribution centered on the encoder's outputs (latents) at training time and has a fixed variance output distribution and trainable priors \cite{theis2017lossy,balle2016end}. It was shown that minimizing the evidence lower bound (ELBO) of this special VAE is equivalent to minimizing jointly the mean square error (MSE) of the reconstruction and the entropy of latents w.r.t the priors \cite{balle2018variational}.
All proposed models differ mainly by the modelling of priors: using either fully-factorized \cite{balle2016end}, zero-mean gaussian \cite{balle2018variational}, gaussian \cite{minen_joint,minnen2020channel} or mixture of gaussian \cite{cheng2020image}, where some methods predict the priors using an autoregressive schema \cite{minen_joint,minnen2020channel,cheng2020image,xie2021enhanced}. 
These neural image codecs were extended to the video compression domain by using two VAEs, one for encoding motion information, another one for encoding residual information in end-to-end video compression \cite{dvc,ssf,aivc}. 


One of the main principle in general compression is to exploit all information available at the decoder to reconstruct the data. However, even though the gradients of the entropy w.r.t latents are available at decoder side, this information remains unused so far in the literature. Some similar work in the literature have tried to improve the performance of codec, for example by using specific parameterization \cite{reducinggap}, or computationally heavy fine-tuning solutions \cite{NEURIPS2020_066f182b, mari2022features,Campos_2019_CVPR_Workshops,LuCZCOXG20}. However, all these methods ignore the gradients of the entropy. 

In this paper, 
we apply the Karush–Kuhn–Tucker (KKT) conditions on the neural codec, which has never been proposed to the best of our knowledge. More specifically, the main contributions of our paper are the following:
\begin{itemize}
    \item We demonstrate that the gradient of the reconstruction error (unavailable at decoder side) is correlated with the gradient of entropy (available at decoder side) w.r.t latents. We bring the theoretical proof that holds in average over the training dataset, and show experimentally that this also holds for each training sample.
    \item We show that it is possible to use the available gradient as a proxy for the unavailable gradient in order to improve the performance of existing neural codecs without re-training.
    \item We report a rate saving between $1-2\%$ at preserved quality, for several neural codecs architectures and without the need for retraining. We also demonstrate that this method can also benefit to network fine-tuning.
\end{itemize}

%% file: problem_statement.tex
\section{NEURAL IMAGE COMPRESSION}
\label{sec:endtoend}

In end-to-end image compression, the encoder $\mathbf{y}=g_a(\mathbf{x};\mathbf{\phi})$ 
 transforms an input image $\mathbf{x} \in \mathbf{R}^{n\times n \times 3}$ into a lower dimensional continuous latent $\mathbf{y} \in \mathbf{R}^{m\times m \times o}$ and quantize it to obtain latent codes $\mathbf{\hat{y}}=\mathbf{Q}(\mathbf{y})$ by element-wise quantization function $\mathbf{Q}(.)$.
Later, the latent codes are compressed losslessly by the entropy coder using factorized entropy model $p_{f}(\mathbf{\hat{y}}|\Psi)$ in \cite{balle2016end}. However, 
if the entropy model of the latent codes $\mathbf{\hat{y}}$ is conditioned with side information to account for the spatial dependencies, the side latent $\mathbf{z}=h_a(\mathbf{y} ;\Phi)$ where $\mathbf{z} \in \mathbf{R}^{k\times k \times f}$ (and side codes $\mathbf{\hat{z}}=\mathbf{Q}(\mathbf{z})$) is also learnt. 
In this setting, the main codes $\mathbf{\hat{y}}$ is encoded with the hyperprior entropy model $p_{h}(\mathbf{\hat{y}}|\mathbf{\hat{z}};\Theta)$, and the side code $\mathbf{\hat{z}}$ is encoded with the factorized entropy model $p_{f}(\mathbf{\hat{z}}|\Psi)$ \cite{balle2018variational,minen_joint}.   
The decoder $\mathbf{\hat{x}}=g_s(\mathbf{\hat{y}};\mathbf{\theta})$ converts back the transmitted latent to the reconstructed image $\mathbf{\hat{x}}$. 

The loss function in \eqref{eq:loss} consists of three parts: the expected bitlength of main codes, the expected bitlength of side codes and the distortion. $d(.,.)$ is any distortion loss such as Mean Square Error (MSE), $\lambda$ is a trade-off parameter to balance between compression ratio and reconstruction quality.

\begin{equation}
   \label{eq:loss}   \mathcal{L}(\mathbf{x})=-log(p_{h}(\mathbf{\hat{y}}|\mathbf{\hat{z}},\Theta)) -log(p_{f}(\mathbf{\hat{z}}|\Psi)) + \lambda d(\mathbf{x},\mathbf{\hat{x}}).
\end{equation}

The training phase aims at finding the optimal network parameters by minimizing \eqref{eq:loss}, for images sampled from given train set $D$:

\begin{equation}
   \label{eq:train}
   \mathbf{\phi}^*, \mathbf{\theta}^*, \Psi^*, \Theta^*,\Phi^*= \mathop{\mathbf{argmin}}_{
    \mathbf{\phi}, \mathbf{\theta}, \Psi, \Theta,\Phi} ( 
    \mathop{\mathbf{E}}_{
    \mathbf{x}\sim D}(\mathcal{L}(\mathbf{x}))  ).
\end{equation}


After offline training, the learned parameters can be deployed to the sender and receiver devices. At test time, the image is transformed to main $\mathbf{\hat{y}}$ and side $\mathbf{\hat{z}}$ codes by trained encoder networks. First, $\mathbf{\hat{z}}$ is encoded by shared factorized entropy model. Then, $\mathbf{\hat{y}}$ is encoded by hyperprior entropy model with input of $\mathbf{\hat{z}}$. At the receiver side, first  $\mathbf{\hat{z}}$ is decoded by shared factorized entropy model. Using $\mathbf{\hat{z}}$ into the hyperprior model, receiver decode $\mathbf{\hat{y}}$ from the bitstream, leading to the reconstructed image using trained decoder network.

\subsection{Fine-tuning solutions}
During offline training, the parameters are optimal in average for the training set, but sub-optimal for a single test image. In order to make image specific RD optimization, the codes can be further fine-tuned at encoding time. In these solutions, the main and side codes are directly optimized for the same loss as in \eqref{eq:loss} while keeping the network parameters fixed as follows: 

\begin{equation}
   \label{eq:finetune}
   \mathbf{\hat{y}}^*,\mathbf{\hat{z}}^*= \mathop{\mathbf{argmin}}_{
    \mathbf{\hat{y}},\mathbf{\hat{z}}} ( 
    \mathcal{L}(\mathbf{x}) ).
\end{equation}

It is reported that, albeit a large increase of encoding time (typically x4000 ), these methods can obtain  a $5\%-15\%$ rate saving compared to the baseline model for the same quality \cite{NEURIPS2020_066f182b,mari2022features,Campos_2019_CVPR_Workshops}.

%% file: proposal.tex
\section{Forgotten Information: Gradient of Entropy}
\label{sec:shiftlatent}
In this section, we describe how the gradients of the entropy on receiver side are used to improve the reconstruction. 
When the receiver decodes side codes, the gradients of side entropy w.r.t side codes can be computed. Similarly, when decoding main codes, the gradients of main entropy w.r.t main codes can also be computed.
Since these gradients are only available after decoding the codes, they seem not useful at first sight. 
Here, we propose to use the gradients through the analysis of the Karush-Kuhn-Tucker (KKT) conditions, and we claim that the first gradient is correlated with the gradient of the main entropy w.r.t side codes and the second one is correlated with the gradient of the reconstruction error w.r.t main codes. Thus, the first one  can be used to decrease the bitlength of the main information and using the second one can decrease the reconstruction error. 

\subsection{Karush-Kuhn-Tucker (KKT) Conditions}
Let us consider the neural codec's loss function in \eqref{eq:loss} as an unconstrained multi-objectives optimization problem, where the objectives are minimum bitlength of main codes, minimum bitlength of side codes and minimum reconstruction error. The optimal solution of multi-objective problem is called Pareto Optimal, when no objective can be improved without degrading the others \cite{miettinen2012nonlinear}. The following remark shows a useful property of a solution of unconstrained multi-objective optimization problems.
\begin{remark}
 \label{rm:kkt}
A solution of the multi-objective optimization problem is Pareto Optimal, if and only if it satisfies Karush Kuhn Tucker (KKT) conditions. More specifically, 
if the aim of the problem is $w^*=\mathop{\mathbf{argmin}}_{w}(\sum_i \alpha_i.\mathcal{L}_i(w))$;
where $\alpha_i \geq 0$, $\sum \alpha_i = 1$ and $\mathcal{L}_i$ is the $i$-th objective to be minimized, the solution $w^*$ is Pareto Optimal if and only if it following condition is met: 
\cite{desideri2012multiple}.
\begin{equation}
\label{eq:moop}
\notag
   \sum_i \alpha_i\nabla_w\mathcal{L}_i(w^*))=0.
\end{equation}
\end{remark}

In plain words, this shows that at the optimal point, all forces driven by gradients cancel each other out and the solution reaches a saddle point. This property is used to test the optimality of the candidate solutions.   
This remark is also valid for end-to-end compression models. The following theorem shows how to use KKT conditions for an end-to-end image compression models.

\begin{theorem}
  \label{Th:th2}
An end-to-end compression model optimized with $\lambda$ trade-off is Pareto Optimal, if and only if the following two conditions are met:
\begin{equation}
\label{eq:kkt1}
\small
   \mathop{\mathbf{E}}_{\substack{\mathbf{x}\sim D }}\left[ \nabla_{\mathbf{\hat{z}}}(-log(p_f(\mathbf{\hat{z}};\Psi)))+ \nabla_{\mathbf{\hat{z}}}(-log(p_h(\mathbf{\hat{y}};\mathbf{\hat{z}},\Theta)))      \right]=0
\end{equation}
\begin{equation}
\small
\label{eq:kkt2}
   \mathop{\mathbf{E}}_{\substack{\mathbf{x}\sim D }}\left[ \nabla_{\mathbf{\hat{y}}}(-log(p_h(\mathbf{\hat{y}};\mathbf{\hat{z}},\Theta)))+ \lambda \nabla_{\mathbf{\hat{y}}}(d(\mathbf{x},g_s(\mathbf{\hat{y}};\mathbf{\theta})))      \right]=0
\end{equation}
\end{theorem}

\begin{proof}
Let's define objectives by
$\mathcal{L}_1(\mathbf{\hat{z}}):=-log(p_f(\mathbf{\hat{z}};\Psi))$,  $\mathcal{L}_2(\mathbf{\hat{y}},\mathbf{\hat{z}}):=-log(p_h(\mathbf{\hat{y}};\mathbf{\hat{z}},\Theta))$ and finally  $\mathcal{L}_3(\mathbf{x},\mathbf{\hat{y}}):=d(\mathbf{x},g_s(\mathbf{\hat{y}};\mathbf{\theta}))$. The coefficients $\alpha_1=\alpha_2=1/(2+\lambda)$ and $\alpha_3=\lambda/(2+\lambda)$ satisfy $\alpha_i \geq 0$ and $\sum \alpha_i = 1$. We can rewrite the  weighed loss with $1/(2+\lambda)$ as follow:
\begin{equation}
\label{eq:loss_ap2}  
   \footnotesize   
\mathop{\mathbf{E}}_{\substack{\mathbf{x}\sim D }}\left[ \frac{\mathcal{L}(\mathbf{x})}{2+\lambda} \right]= \mathop{\mathbf{E}}_{\substack{\mathbf{x}\sim D }}\left[\alpha_1\mathcal{L}_1(\mathbf{\hat{z}})+\alpha_2\mathcal{L}_2(\mathbf{\hat{y}},\mathbf{\hat{z}})+\alpha_3\mathcal{L}_3(\mathbf{x},\mathbf{\hat{y}}) \right].
\end{equation}

Since it is an unconstrained multi-objective optimization problem, RHS should meet $\mathop{\mathbf{E}}_{\substack{\mathbf{x}\sim D }}\left[\sum_i \alpha_i\nabla_w\mathcal{L}_i(w^*)) \right]=0$ (KKT conditions). When we parameterize \eqref{eq:loss_ap2} with $\mathbf{\hat{z}}$, $\mathcal{L}_3$ is independent of $\mathbf{\hat{z}}$, thus $\nabla_{\mathbf{\hat{z}}}\mathcal{L}_3(\mathbf{x},\mathbf{\hat{y}})=0$. The first KKT conditions w.r.t $\mathbf{\hat{z}}$ can be equivalent to \eqref{eq:kkt1}:

\begin{equation}
   \label{eq:loss_ap3}
   \footnotesize
     \mathop{\mathbf{E}}_{\mathbf{x}\sim D} \left[ \alpha_1\nabla_{\mathbf{\hat{z}}}\mathcal{L}_1(\mathbf{\hat{z}})+\alpha_2\nabla_{\mathbf{\hat{z}}}\mathcal{L}_2(\mathbf{\hat{y}},\mathbf{\hat{z}})    \right]=0.
\end{equation}

When we parameterize \eqref{eq:loss_ap2} with $\mathbf{\hat{y}}$, $\mathcal{L}_1$ is independent of  $\mathbf{\hat{y}}$, thus $\nabla_{\mathbf{\hat{y}}}\mathcal{L}_1(\mathbf{\hat{z}})=0$. The second KKT conditions w.r.t $\mathbf{\hat{y}}$ can be equivalent to \eqref{eq:kkt2}:

\begin{equation}
   \label{eq:loss_ap4}
   \footnotesize
    \mathop{\mathbf{E}}_{\mathbf{x}\sim D} \left[ \alpha_2\nabla_{\mathbf{\hat{y}}}\mathcal{L}_2(\mathbf{\hat{y}},\mathbf{\hat{z}})+ \alpha_3\nabla_{\mathbf{\hat{y}}}\mathcal{L}_3(\mathbf{x},\mathbf{\hat{y}})   \right]=0.     
\end{equation}
\end{proof}

The proof is straight-forward, but the result of this theorem is significant. One can interpret the theorem that if an existing end-to-end models is optimal (at least in terms of training {performance}, not in terms of compression performance), gradient of side information’s entropy w.r.t side latents $\nabla_{\mathbf{\hat{z}}}(-log(p_f(\mathbf{\hat{z}};\Psi)))$ and gradient of main information’s entropy w.r.t side latents $ \nabla_{\mathbf{\hat{z}}}(-log(p_h(\mathbf{\hat{y}};\mathbf{\hat{z}},\Theta)))$ cancel themselves out in expectation. It is the same for the second condition as well. We can claim that main information’s entropy w.r.t main latents $\nabla_{\mathbf{\hat{y}}}(-log(p_h(\mathbf{\hat{y}};\mathbf{\hat{z}},\Theta)))$ and weighted gradient of reconstruction error w.r.t main latents $\lambda \nabla_{\mathbf{\hat{y}}}(d(\mathbf{x},g_s(\mathbf{\hat{y}};\mathbf{\theta})))$ cancels themselves out in expectation.


 \begin{figure*}[htb]
 \begin{center}
\includegraphics[width=0.99\linewidth]{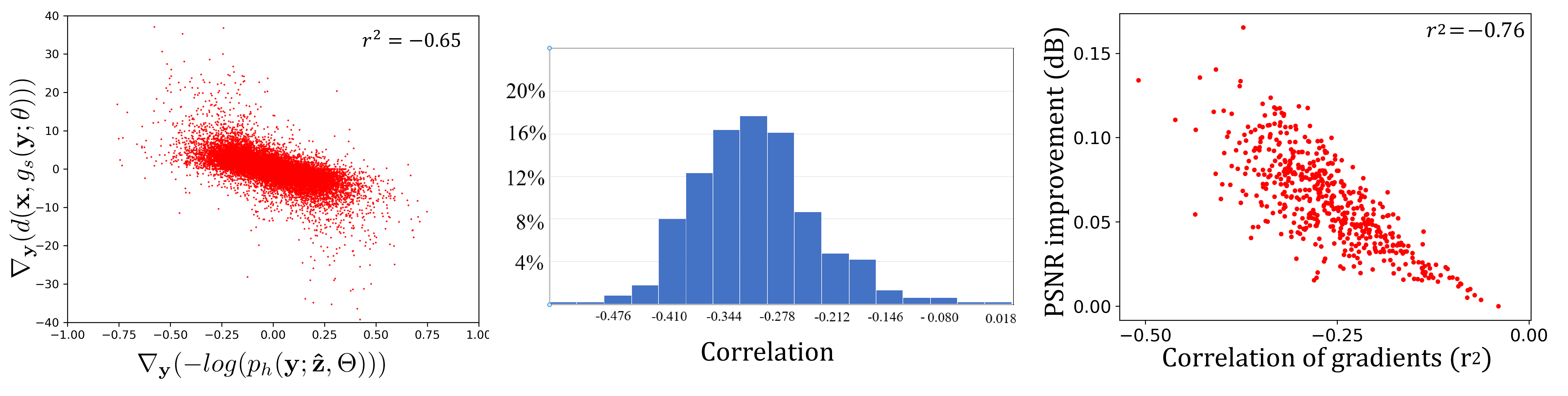}
\vskip -0.15in
\caption{a) Correlation between gradients of the entropy and the reconstruction error. b) Histogram of correlation between gradients w.r.t main latents. c) Correlation between improvement on reconstruction quality and correlation of gradients. }
\label{fig:corr}
\end{center}
\vskip -0.15in
\end{figure*}

If these conditions were valid for every single image, a given test image's available gradient would be the opposite of the unavailable gradient, thus they would have $-1$ correlation coefficient. These conditions are valid in average on the training set. However, we claim that this relationship still holds somehow for any single test image. Let us verify experimentally that this assumption is true. Figure~\ref{fig:corr}a shows the correlation between a sample image's gradients w.r.t the main latents and gradients w.r.t reconstruction error. In this specific case, the correlation coefficient of the gradients is $-0.65$.  
According to our test on several neural codecs, we have found quite a strong correlation between gradients 
as it is shown on histogram in Figure~\ref{fig:corr}b. 

\subsection{Latent Shift w.r.t Gradients} 
In this section, we explain how the relationship between these gradients can be used at decoding side to improve the performance of the codec. The main intuition is to shift the latents in the direction of the known gradient, as it should have a positive influence on the bitrate of the side latent and the reconstructed image with negligible cost in decoding time. 

By definition of gradient based optimization, $\mathbf{\hat{z}}$ needs to take a step in the negative direction of $\nabla_{\mathbf{\hat{z}}}(-log(p_h(\mathbf{\hat{y}};\mathbf{\hat{z}},\Theta)))$ in order to decrease the main information bitlength that is $-log(p_h(\mathbf{\hat{y}};\mathbf{\hat{z}},\Theta)$. However $\nabla_{\mathbf{\hat{z}}}(-log(p_h(\mathbf{\hat{y}};\mathbf{\hat{z}},\Theta)))$ is not available before decoding $\mathbf{\hat{y}}$, but the correlated gradient $\nabla_{\mathbf{\hat{z}}}(-log(p_f(\mathbf{\hat{z}};\Psi)))$ is known after decoding $\mathbf{\hat{z}}$. We claim that there is a step size $\rho_f^*$ that decrease the bitlength of main information. 
$\rho_f^*$ can be obtained by brute force or any optimization method to find the optimal such that:
\begin{equation}
\label{eq:opt_side}
\footnotesize
   \rho_f^*=\mathop{\mathbf{argmin}}_{\rho_f}(-log(p_h(\mathbf{\hat{y}};\mathbf{\hat{z}}+\rho_f\nabla_{\mathbf{\hat{z}}}(-log(p_f(\mathbf{\hat{z}};\Psi)))),\Theta)).
\end{equation}

For the second condition, $\mathbf{\hat{y}}$ needs to take a step in the negative direction of $\nabla_{\mathbf{\hat{y}}}(d(\mathbf{x},g_s(\mathbf{\hat{y}};\mathbf{\theta})))$ in order to decrease reconstruction error $d(\mathbf{x},g_s(\mathbf{\hat{y}};\mathbf{\theta}))$. Since $\nabla_{\mathbf{\hat{y}}}(d(\mathbf{x},g_s(\mathbf{\hat{y}};\mathbf{\theta})))$ is not available at decoding, we use the correlated gradient $\nabla_{\mathbf{\hat{y}}}(-log(p_h(\mathbf{\hat{y}};\mathbf{\hat{z}},\Theta)))$  which is known after decoding $\mathbf{\hat{y}}$ and $\mathbf{\hat{z}}$.  We claim that there is a step size $\rho_h^*$ that decrease the reconstruction error. 
$\rho_h^*$ can be found by brute-force search or any optimization method to find  the optimal such that:
\begin{equation}
\label{eq:opt_main}
\footnotesize
   \rho_h^*=\mathop{\mathbf{argmin}}_{\rho_h}(d(\mathbf{x},g_s(\mathbf{\hat{y}}+\rho_h\nabla_{\mathbf{\hat{y}}}(-log(p_h(\mathbf{\hat{y}};\mathbf{\hat{z}},\Theta)));\mathbf{\theta}))).
\end{equation}
In summary, our proposal can be seen as shifting the side latent by $\mathbf{\hat{z}} \leftarrow \mathbf{\hat{z}}+\rho_f^*\nabla_{\mathbf{\hat{z}}}(-log(p_f(\mathbf{\hat{z}};\Psi)))$ after decoding $\mathbf{\hat{z}}$ and shifting the main latent by $\mathbf{\hat{y}} \leftarrow \mathbf{\hat{y}}+\rho_h^*\nabla_{\mathbf{\hat{y}}}(-log(p_h(\mathbf{\hat{y}};\mathbf{\hat{z}},\Theta)))$ after decoding $\mathbf{\hat{y}}$ while the best step sizes $\rho_f^*, \rho_h^* \in \mathbf{R}$ are to be found at encoding time and signaled in the bitstream. It is noted that during the encoding, the main latents $\mathbf{\hat{y}}$ is encoded w.r.t entropy model predicted by the shifted side latent.
We show that the shifted side latents shortens the main information's bitlength and shifted main latents results to obtain better reconstruction performance.

%% file: results.tex
\section{Experimental Results}
\label{sec:exp}
We use CompressAI library \cite{compressai} to test our contribution on $5$ pre-trained neural image codecs named \textbf{bmshj2018-factorized} in \cite{balle2016end}, \textbf{mbt2018-mean} and \textbf{mbt2018} in \cite{minen_joint},  \textbf{cheng2020-attn} in \cite{cheng2020image} and  \textbf{invcompress} in \cite{xie2021enhanced}. 
For the evaluation, we use Kodak dataset \cite{eastman_kodak_kodak_nodate} and Clic-2021 Challenge's Professional dataset \cite{CLIC}. 

\begin{table}[!htbp]
\centering
\caption{Average BD-Rate gains of \textbf{Latent Shift} for different baseline image codecs on 2 image datasets.}
\label{tab:image}
\footnotesize
\begin{tabular}{lcc}
\toprule
Baseline Codec & Kodak Test Set & Clic-2021 Test Set \\
\midrule
bmshj2018-factorized & -0.49\%  & -0.69\%    \\
mbt2018-mean &  -1.27\%  & -1.21\%   \\
mbt2018  & -1.44\% & -1.71\%  \\
cheng2020-attn & -0.46\%  & -0.72\%  \\
InvCompress &	-0.55\%	 & -0.63\% \\
\bottomrule
\end{tabular}
\end{table}

\begin{table*}[h]
\centering
\caption{Computational complexity and \textbf{Latent Shift} performance over Finetuning solutions}
\label{tab:finetune}
\footnotesize
\begin{tabular}{lllcccc}
\toprule
Model&	Encoding Time&	Decoding Time&	Bd-Rate (Kodak)	&Corr. (Kodak)&	Bd-Rate (Clic)&	Corr.  (Clic) \\
\midrule
bmshj2018 baseline  & x1  & +0.0\% & 0.0\% &  & 0.0\% &   \\
Only \textbf{Latent Shift} &  x12&	+0.7\%&	-0.49\% &	-0.108&	-0.69\%	&-0.0952  \\
Only FineTuning  & x4800&	+0.0\%	&-6.52\%	&	&-6.73\%	&  \\
FineTuning + \textbf{Latent Shift} &x4812&	+0.7\%&	-7.88\%&	-0.165&	-8.06\%&	-0.1578  \\
\midrule
mbt2018-mean baseline & x1  & +0.0\%  &	0.0\%	&  &	0.0\%&  \\
Only \textbf{Latent Shift} & x10.1 &	+0.7\%	&-1.27\%	&-0.1805	&-1.21\% &	-0.1465   \\
Only FineTuning  & x4380 &	+0.0\% &	-5.77\%	& 
 &	-5.53\%	&  \\
FineTuning + \textbf{Latent Shift} & x4390 &	+0.7\%&	-7.47\% &	-0.2212	&-7.16\% &	-0.1796 \\
\bottomrule
\end{tabular}
\end{table*}

The codecs are taken off the shelf and are not retrained. The rate is calculated from the final length of the compressed data and RGB PSNR is used for distortion metric. We assess the performance using the bd-rate \cite{bjontegaard2001calculation} on the rate range of 0.1bpp-1.6bpp used in compressAI.

In Table~ \ref{tab:image}, we show the results of our \textbf{Latent Shift} method for different datasets and codecs. Even though the bd-rate gains are moderate, it constantly saves bit from wide variety of the neural architectures. The \textbf{Latent Shift}'s performance depends on how the gradients are correlated. To assess this correlation, we show the scatter plot between \textbf{Latent Shift}'s individual gains in dB and actual correlation coefficient between gradients for all test images and for all quality level in Figure \ref{fig:corr}c. We found a correlation of $-0.76$ where this correlation is independent from datasets, reconstruction quality or model. 

We have evaluated complexity of \textbf{Latent Shift} over the selected two baseline models with and without fine-tuning solutions in Table~ \ref{tab:finetune}. Since we fine-tune the latents for 1000 iterations during encoding time, it needs 1000 forward pass and 1000 backward gradient calculations that results to around 4000 times more complexity without parallelization. On the other hand, our \textbf{Latent Shift} needs to test the best step size ($\rho_f,\rho_h$) out of 8 predefined candidates. Thus \eqref{eq:opt_side} and \eqref{eq:opt_main} need 8 forward pass and one single gradient calculation whose closed form solution is available. It explains why \textbf{Latent Shift} gives around 10 times encoding complexity without parallelization. However, both approaches have negligible decoder complexity.  We report the complexities on single thread as opposite to the previous researches with massive paralellization in order to understand their necessary number of clock cycle as a proxy of energy consumption. 

An important result can be seen when \textbf{Latent Shift} is applied after fine-tuning solutions. Since fine-tuning solution applies the minimization of the loss without averaging images in train set as in \eqref{eq:finetune}, KKT conditions get strength and existed correlations between gradients are increased. For instance the average gradients correlation for Kodak set becomes $-0.2212$ after fine-tuning while it was $-0.1805$ in \textbf{mbt2018-mean} codec. The improvement on the correlations results improvement of the performance as it can be seen that fine-tuning solution only has $5.77\%$ rate saving while together with \textbf{Latent Shift}, it increases to $-7.47\%$ for Kodak dataset. 

%% file: conclusion.tex
\section{Conclusion}
\label{sec:conclusion}
\vskip -0.01in
In this work, we have proposed \textbf{Latent Shift} to improve further the latent representation of compressive variational auto-encoders (VAE). We have demonstrated the correlation between the entropy gradient and the distortion gradient, and show that this can be used to bring significant gains on top of several state-of-the-art compressive auto-encoders, without any need of retraining.
From these results, several improvements can be foreseen. First, the correlation of the gradients depends on the training set according to the definition of the KKT conditions. However, even though different models use the same training set and training procedure, their gradients' correlation coefficient may be different. Thus, it would be interesting to explore the connection of certain type of model architectures with the correlation of the gradients.      